\documentclass{amsart}
\usepackage{amsmath, amsthm}
\usepackage{amssymb, amscd, latexsym, graphicx, psfrag}
\usepackage[all,cmtip]{xy}
\usepackage{hyperref}
\hypersetup{colorlinks=true, linkcolor=black}
\numberwithin{equation}{section}
\newtheorem{Prop}{Proposition}
\newtheorem{Thm}{Theorem}
\newtheorem{Lemma}{Lemma}
\newtheorem{Cor}{Corollary}
\theoremstyle{definition}
\newtheorem{Def}{Definition}
\theoremstyle{remark}
\newtheorem{Rem}{Remark}
\theoremstyle{definition}

\title{Frobenius structures on double Hurwitz spaces}
\author{Stefano Romano}
\date{}
\begin{document}
\begin{abstract} We construct Frobenius structures of ``dual type" on the moduli space of ramified coverings of $\mathbb{P}^1$ with given ramification type over two points, generalizing Dubrovin's construction of \cite{Du-Diff, Du-2dtft}. A complete hierarchy of hydrodynamic type is obtained from the corresponding deformed flat connection. This provides a suitable framework for the Whitham theory of an enlarged class of integrable hierarchies; we treat as examples the q-deformed Gelfand-Dickey hierarchy and the sine-Gordon equation, and compute the corresponding solutions of the WDVV equations.
\end{abstract}
\maketitle
\section*{Introduction}
Since the discovery in the seventies \cite{Nov-Per, DuMaNo, Kr0} of algebro-geometric solutions of soliton equations, the notion of ``spectral curve" (roughly speaking, a Riemann surface with a marked meromorphic function) proved to play a central role in the theory of integrable systems. For any commuting hierarchy of PDEs of KdV type, the algebro-geometric machinery produces a rich supply of exact quasi-periodic solutions built out of theta functions on the Jacobian of a spectral curve. These solutions come in families, parametrized by a suitable moduli space of spectral curves; we loosely refer to the latter as a ``Hurwitz space".\\
\indent The existence of Frobenius structures \cite{Du-2dtft} on Hurwitz spaces may be viewed as the cornerstone of the connection between Whitham theory of slow modulation of soliton lattices and topological field theory \cite{Du-With, Kr01, Kr1}: on one hand, Frobenius structures on moduli spaces of quasi-periodic solutions are directly inherited from the bi-hamiltonian structure of the original hierarchy through a procedure of ``Whitham-averaging" over the invariant tori \cite{DuNo-Hydro}; on the other hand, the notion of Frobenius manifold itself represents the geometric axiomatization of the mathematical structure of topological field theories.\\
\indent According to the standard definition, Hurwitz spaces are moduli spaces of ramified coverings of $\mathbb{P}^1$ of a given genus and given ramification type over the point $\infty\in\mathbb{P}^1$ (plus possibly some additional discrete data). A complete construction of Frobenius structures on these spaces was established by Dubrovin in \cite{Du-Diff, Du-2dtft} (and further developed by Shramchenko \cite{Sch1, Sch2}); it involves the choice of a ``quasi-momentum differential", which is the analogue of Saito's primitive form in the framework of unfoldings of isolated hypersurface singularities \cite{Sa-Prim}. In our context, it is a certain moduli-dependent Abelian differential on the spectral curve (possibly multi-valued) with singularities only at the points lying over $\infty$.\\

In this work we consider a more general class of Hurwitz spaces, allowing the ramification of the cover to be assigned over two points, $0$ and $\infty$; we call them ``double" Hurwitz spaces. They are closely related to the so-called ``double-ramification cycles", which in recent years proved to be important characters in the study of the tautological ring of $\overline{\mathcal{M}}_{g,n}$. Our interest here is however confined to the Frobenius structures on these spaces and the related Whitham-type hierarchies.\\
The fact that a double Hurwitz space carries some sort of Frobenius structure does not come as a surprise. It was shown by Krichever \cite{Kr1} that an exact solution of WDVV can be associated to a moduli space of spectral curves in a very broad sense; in his setting, the meromorphic function on the spectral curve needs not even be single-valued, and can in fact be replaced by an Abelian integral. The present work can be seen as a detailed study of this kind of weak Frobenius structures on generalized Hurwitz spaces in a fairly simple case\footnote{We essentially consider Abelian integrals of differentials of the third kind with integer residues at the marked points (the covering map being identified with the exponential of the Abelian integral, and thus being single-valued). Although we will not treat it here, let us mention that the methods we present admit a straightforward generalization to the case of non-integer residues. The latter turns out to have interesting applications in the context of equivariant Gromov-Witten theory: concrete examples suggest that for certain toric varieties a suitable mirror model is given by a multivalued Ginzburg-Landau superpotential, with periods related to equivariant parameters. This topic will be investigated in a subsequent publication with A. Brini, G. Carlet and P. Rossi.}.\\
\indent Our approach is a direct adaptation to the double ramification setting of that of Dubrovin in \cite{Du-Diff}: we use the language of horizontal differentials and their universal pairing to define the relevant metrics on the moduli space in canonical coordinates. Let us summarize our main results:
\begin{itemize}
\item We show that a double Hurwitz space is endowed with a Frobenius structure of the type arising from Dubrovin's almost-duality \cite{Du-Alm}; this roughly speaking means that the unit vector field and the Euler vector field coincide in this structure. 
\item From the deformed almost-dual flat connection we extract a complete hierarchy of hydrodynamic type, naturally represented in the ``invariant" form of Flaschka,  Forest and McLaughlin \cite{FFM}.
\item We observe that whenever one of the two assigned ramification profiles is trivial, a new compatible flat metric appears on the Hurwitz space. This recovers Dubrovin's construction for ``single" Hurwitz spaces, but also prompts an additional remark: the Hurwitz spaces without tangency constraints carry a \emph{tri-hamiltonian} Frobenius structure in the sense of \cite{Rom}.
\item We provide explicit formulas for the flat coordinates of the metrics involved, expressed in terms of pairings of horizontal differentials. This yields a practically efficient way of computing the corresponding prepotentials.
\end{itemize}
Again, all the Frobenius structure appearing depend upon the choice of a quasi-momentum differential. One of the main novelties of the double-ramification setting is that singularities of the quasi-momentum at the zeros of the covering map become admissible.\\

Our construction is motivated by several instances of integrable hierarchies whose dispersionless limit/Whitham theory appears to be best treated in the present setting. An example appeared in \cite{BrCaRo}, where the dispersionless limit of the Ablowitz-Ladik hierarchy is shown to be described by a Hurwitz space with a non-admissible quasi-momentum differential. We consider a similar case here, namely the dispersionless limit of the q-deformed Gelfand-Dickey hierarchy \cite{Fr-Def, FrReSTS}. As a second example we study the Whitham-averaging of the sine-Gordon equation over the space of $g$-phase solutions, which is known  to be described by a double Hurwitz space of hyperelliptic curves \cite{EFMM}.

\subsection*{Acknowledgments} I would like to thank my advisor B. Dubrovin for drawing my attention to examples 4.1 and 4.2, which were the main source of inspiration for this work, and for many valuable suggestions and remarks. This work has been partially supported by the project FroM-PDE funded  by the European Research Council through the Advanced Investigator Grant Scheme.
\section{Preliminaries}
\subsection{Frobenius structures}
In this work we will make use of a few different notions of Frobenius structure; we collect them in the following
\begin{Def} Let $M$  be an $n$-dimensional complex manifold. A \emph{Frobenius structure} on $M$ is the data of
\begin{itemize}
\item A flat metric $\eta$ on $M$. By this we mean a holomorphic symmetric non-degenerate $(0, 2)$-tensor with flat Levi-Civita connection.
\item A commutative unital algebra structure on the tangent spaces of $M$, depending holomorphically on the point and compatible with $\eta$ in the sense that
$$
\eta(X\cdot Y, Z)=\eta(X, Y\cdot Z)\qquad \text{for all vector fields}\;\; X, Y, Z
$$
where $\cdot$ denotes the product of the algebra.
\end{itemize}
such that the pencil of affine connections
\begin{equation}\label{defEC}
\nabla^{(z)}_X Y\doteq \nabla_X Y+zX\cdot Y\qquad z\in\mathbb{C}
\end{equation}
is identically flat in the complex parameter $z$. Here $X, Y$ are arbitrary vector fields on $M$ and $\nabla$ is the Levi-Civita connection of $\eta$. We distinguish three types of Frobenius structure:\vspace{10pt}
\\
{\bf(D)}\hspace{10pt}The Frobenius structure is \emph{dual-type} if there is a constant $d$ such that
\begin{equation}\label{dualtype}
\nabla df=0\;\;\Rightarrow\;\;\partial_e f=\frac{1-d}{2}f
\end{equation}
for all functions $f$ on $M$.\vspace{5pt}\\
{\bf(C)}\hspace{10pt} The Frobenius structure is \emph{conformal} if
\begin{itemize}
\item[(i)] The unit vector field $\partial_e$ is flat with respect to $\nabla$.
\item[(ii)] There exists an \emph{Euler vector field} $\partial_E$ satisfying the following:
\begin{itemize}
\item The endomorphism $\nabla\partial_E$ of the tangent spaces is covariantly constant and diagonalizable.
\item Define the \emph{grading operator} $\hat{\mu}$ as the traceless part of $\;-\nabla\partial_E$,
\begin{equation}\label{gradingOP}
\hat{\mu}\doteq \left(1-\frac{d}{2}\right) I-\nabla\partial_E,\hspace{25pt}d= 2-2\; \text{tr}(\nabla\partial_E)
\end{equation}
and $\mathcal{U}$ as the operator of multiplication by the Euler vector field: 
$$
\mathcal{U}X\doteq \partial_E\cdot X
$$
Then the flat pencil \eqref{defEC} extends to a \emph{flat} meromorphic connection on $M\times\mathbb{P}^1$ via the formulas
$$
\nabla^{(z)}\frac{\partial}{\partial z}=0
$$
$$
\nabla^{(z)}_{\partial/\partial z} X= \frac{\partial}{\partial z}X+\mathcal{U} X-\frac{1}{z}\hat{\mu}X
$$
\end{itemize}
\end{itemize}
\vspace{5pt}
{\bf(T)}\hspace{10pt} The Frobenius structure is \emph{tri-hamiltonian} if it is conformal, $n$ is even and the grading operator \eqref{gradingOP} has only two eigenvalues $\pm d/2$ with multiplicity $n/2$.\vspace{10pt}\\
For each of the three definitions, the constant $d$ is called the \emph{charge}.
\end{Def}
Part (C) coincides with Dubrovin's original definition of Frobenius manifold \cite{Du-2dtft}. One of the main consequences of the conformality axiom (ii) is the existence of a \emph{flat pencil of metrics} on $M$, formed by $\eta$ together with the \emph{intersection form}
$$
g(X, Y)\doteq \eta(X, \mathcal{U}^{-1 }Y)
$$
The latter is only defined outside the locus of degeneracy of $\mathcal{U}$ (the so-called \emph{discriminant}). The notion (D) of dual-type Frobenius structure was introduced by Dubrovin in the context of almost-duality \cite{Du-Alm}: he showed that any conformal Frobenius manifold carries an additional dual-type structure, in which the flat metric is $g$ and the product is
\begin{equation}\label{star}
X\star Y\doteq \mathcal{U}^{-1}X\cdot Y
\end{equation}
\begin{Rem} For $d=1$ \eqref{dualtype} means that the unit is flat; thus a dual-type structure of charge one can be equivalently described as non-conformal Frobenius (i.e. without Euler vector field). Relevant examples arise in the context of equivariant Gromov-Witten theory.
\end{Rem}
The notion (T) of tri-hamiltonian Frobenius structure was introduced by Pavlov and Tsarev \cite{PT} in the semisimple setting, and treated in general in \cite{Rom}. Its interest lies in the fact that on a tri-hamiltonian Frobenius manifold the third metric
$$
\tilde{\eta}(X, Y)\doteq\eta(X, \mathcal{U}^{-2 }Y)
$$
is flat and forms flat pencils with both $\eta$ and $g$; moreover, the product
\begin{equation}\label{ast}
X\ast Y \doteq \mathcal{U}^{-2}X\cdot Y
\end{equation}
induces together with $\tilde{\eta}$ a third Frobenius structure on $M$. To each of the three structures one can associate a prepotential (i.e. a solution of the WDVV equations) in the standard way: the third derivatives of the prepotential with respect to flat coordinates are the totally covariant structure constants of the corresponding product.\\

Finally, recall that a Frobenius structure is \emph{semisimple} if at a generic point the algebra structure on the tangent space contains no nilpotents. Near a semisimple point one may locally construct \emph{canonical coordinates} $v_1, \dots, v_n$  reducing the multiplication table to
$$
\frac{\partial}{\partial v_i}\cdot \frac{\partial}{\partial v_j}=\delta_{ij}\frac{\partial}{\partial v_i}
$$
and (consequently) the metric to diagonal form. In the following we will always deal with semisimple structures and work directly in a system of canonical coordinates.
\subsection{Double Hurwitz spaces}
Let $\mu=(m_1, \dots, m_K)$, $\nu=(n_1, \dots, n_L)$ be two collections of positive integers with equal sum.
\begin{Def}
A \emph{Hurwitz cover} (or \emph{spectral curve}) of type $g, \mu, \nu$ is a pair $(C_g, \lambda)$, where
\begin{itemize}
\item $C_g$ is a smooth genus $g$ curve with $K+L$ distinct ordered marked points $z_1, \dots, z_K, p_1, \dots, p_L$ and a marked symplectic homology basis $a_1, \dots a_g$, $b_1, \dots, b_g \in H_2(C_g, \mathbb{Z})$.
\item $\lambda:C_g\to \mathbb{P}^1$ is a meromorphic function with zeros at the $z_a$ and poles at the $p_b$, with respective multiplicities $m_a, n_b$. That is, we have
$$
(\lambda)=\sum_{a=1}^Km_az_a-\sum_{b=1}^L n_b p_b
$$
as divisors on $C_g$.
\end{itemize}
The \emph{double Hurwitz space} $\mathcal{H}_{g, \mu, \nu}$ is the moduli space of Hurwitz covers of type $g, \mu, \nu$, where equivalence of covers is given by commutative diagrams
$$\xymatrix{
C_g \ar[d]_\lambda \ar[r]^\sim &C_g'\ar[ld]^{\lambda'}\\
\mathbb{P}^1 }$$
\end{Def}
Note that the base $\mathbb{P}^1$ is parametrized: any automorphism of $\mathbb{P}^1$ acts nontrivially on the Hurwitz space, as long as it preserves the type of the cover.\\
\begin{Rem}
Consider in particular the $\mathbb{C}_*$-action given by
\begin{equation}\label{resc}
(C_g, \lambda)\mapsto (C_g, a\lambda)\qquad a\in\mathbb{C_*}
\end{equation}
which always preserves the type of the cover. The quotient of the moduli space by this action is closely related to the so-called
\emph{double ramification cycle} $DR_{g, \mu, \nu}$, which is the subvariety of moduli space $\mathcal{M}_{g, K+L}$ of smooth genus $g$ curves with $K+L$ marked points defined by the requirement that the divisor
$$
\delta_{\mu, \nu}\doteq\sum_{a=1}^Km_az_a-\sum_{b=1}^L n_b p_b
$$
be principal:
\begin{align*}
DR_{g, \mu,\nu}=	&\Big\{(C_g, z_a, p_b)\in\mathcal{M}_{g, K+L};\; \exists\; \lambda: C_g\to \mathbb{P}^1\;\;\text{such that}\;\; (\lambda)=\delta_{\mu, \nu}\Big\}
\end{align*}
Imposing $\delta_{\mu,\nu}$ to be principal amounts to $g$ equations on $\mathcal{M}_{g, K+L}$, so that $$\text{dim}\,DR_{g, \mu, \nu}= \text{dim}\mathcal{M}_{g, K+L}-g=2g-3+K+L$$
We have an obvious morphism
$$
\pi:\mathcal{H}_{g, \mu, \nu}\to DR_{g, \mu, \nu}
$$
which forgets the map $\lambda$ (equivalently, quotients by the action \eqref{resc}) and the marked homology basis. Since the second datum is discrete, the fiber of $\pi$ is 1-dimensional and we have
$$
\text{dim}\mathcal{M}_{g, \mu, \nu}= 2g-2+K+L\doteq n
$$
\end{Rem}
\vspace{15pt}
The space $\mathcal{H}_{g, \mu, \nu}$ may in general consist of several connected components. In the following we work on some chosen connected component of $\mathcal{H}_{g, \mu, \nu}$. Moreover, we restrict to the open dense subset consisting of ``generic" covers, in the sense that all zeros of $d\lambda$ are assumed to be simple away from the marked points. We denote by $\mathcal{M}_{g, \mu, \nu}$ the resulting space.\\

By Riemann-Hurwitz formula, a generic cover has $n=2g-2+K+L$ simple branch points, which equals the dimension of the space. The morphism sending $(C_g, \lambda)$ to the disordered set of critical values of $\lambda$ is unramified of finite degree. Hence the (arbitrarily ordered) branch points $u_1, \dots, u_n$ of $\lambda$ provide local coordinates on $\mathcal{M}_{g, \mu, \nu}$.\\
\indent Let $\mathcal{C}_{g, \mu, \nu}\to\mathcal{M}_{g, \mu, \nu}$ be the universal curve, i.e. the bundle over $\mathcal{M}_{g, \mu, \nu}$ whose fiber at the point $u\in \mathcal{M}_{g, \mu, \nu}$ is the $C_g(u)$. The functions $\lambda$ on the fibers glue to a map (which we denote by the same letter)
$$
\lambda: \mathcal{C}_{g, \mu, \nu}\to\mathbb{P}^1
$$
We identify nearby fibers along the level sets
\begin{equation}\label{horiz}
\lambda=\text{const}
\end{equation}
and so obtain a canonical lifting of vector fields from the base to the universal curve.\\
The marked points $z_1, \dots, z_K, p_1, \dots, p_L$ provide global horizontal sections of $\mathcal{C}_{g, \mu, \nu}$. In general, level sets of $\lambda$ fail to be transversal to the fibers at the ramification locus of $\lambda$: let $C_g$ be the fiber over $u$ and $P_i\in C_g$ be the $i$-th ramification point: 
\begin{equation}\label{ramifPi}
d\lambda(P_i)=0, \qquad \lambda(P_i)=u_i\qquad i=1, \dots, n
\end{equation}
Then the horizontal lift $\partial/\partial u_i$ is tangent to $C_g(u)$ at $P_i$. Another way to view this is to observe that near the point $P_i$  the natural local parameter on the curve is
$$
\tau=(\lambda - u_i)^{1/2}\qquad \text{near}\;\; P_i
$$
which is not horizontal as it depends explicitly on $u_i$.
\section{Horizontal differentials and flat metrics}
To simplify notation, we drop the subscripts $g, \mu, \nu$ everywhere from now on.
\begin{Def}Let $(C, \lambda)$ be a spectral curve belonging to $\mathcal{M}$. An \emph{admissible differential} on $C$ is a holomorphic differential $\Omega$ on the universal cover of $C\setminus \{ z_1\dots, z_K, p_1, \dots, p_L\}$ such that for all $\gamma\in H_1(C\setminus \{z_a, p_b\}, \mathbb{Z})$ we have
\begin{equation}\label{jumps}
\Omega(P+\gamma)-\Omega(P)=\sum_{k\in\mathbb{Z}}c_k(\gamma)\lambda^{k}d\lambda\qquad P\in C
\end{equation}
with only a finite number of $c_k$ non-zero.
\end{Def}
The (infinite dimensional) vector space of admissible differentials on a spectral curve induces a (infinite rank) vector bundle $\mathcal{A}\to\mathcal{M}$ on the Hurwitz space. The \emph{horizontal differentials} are certain sections of this bundle. The horizontality condition is as follows: let $\Omega=\Omega(u)$ be a local section of $\mathcal{A}$. We view $\Omega$ as a multivalued 1-form on the universal curve $\mathcal{C}$, and denote by $\partial_i\Omega$ its Lie-derivative with respect to $\partial/\partial u_i$. Recall that the latter is defined as the horizontal lift to the universal curve of the coordinate vector field $\partial/\partial u_i$ on the base.
\begin{Def}
A local section $\Omega=\Omega(u)$ of $\mathcal{A}$ is \emph{horizontal} if, for all $i=1, \dots, n$, the restriction of $\partial_i\Omega$  to each fiber of the universal curve is an Abelian differential of the second kind, with only a double pole at $P_i$ and zero $a$-periods.
\end{Def}
To rephrase more concretely the definition, we spell out the analytic properties of a generic admissible differential $\Omega$ on a given spectral curve:
\begin{align}\label{ap}
&\Omega =\sum_{k<-1}B_{a,k}\tau^kd\tau+\frac{1}{m_a}\sum_{l\in\mathbb{Z}}R_{a,l}d(\lambda^l\log\lambda)+\text{reg}\\
&\hspace{170pt}\text{at}\; z_a,\; \tau=\lambda^{1/m_a}, \;a=1, \dots, K\nonumber\\
&\Omega =\sum_{k<-1}C_{b,k}\tau^kd\tau+\frac{1}{n_a}\sum_{l\in\mathbb{Z}}S_{b,l}d(\lambda^l\log\lambda)+\text{reg}\nonumber\\ &\hspace{170pt}\text{at}\; p_b,\; \tau=\lambda^{-1/n_b}, \;b=1, \dots, L\nonumber\\
&\Omega(P+a_\alpha)-\Omega(P)=\sum_{k\in\mathbb{Z}}D_{\alpha, k}\lambda^kd\lambda\equiv dD_\alpha(\lambda)\hspace{60pt} \alpha=1, \dots, g\nonumber\\
&\Omega(P+b_\alpha)-\Omega(P)=\sum_{k\in\mathbb{Z}}E_{\alpha, k}\lambda^kd\lambda\equiv dE_\alpha(\lambda)\hspace{62pt} \alpha=1, \dots, g\nonumber\\
&\oint_{a_\alpha}\Omega=A_\alpha\hspace{206pt}\alpha=1, \dots, g\nonumber
\end{align}
\begin{Rem}
Some care is needed in the definition of the periods of a multi-valued differential. We let all the periods be computed along paths starting at some fixed point $P_0$ and ending at the same point.  When working in families, we assume that $P_0$ belongs to an horizontal section $\lambda(P_0)=\text{const}$ of the universal curve. Then the $a$- and $b$-periods of a multivalued differential of the form \eqref{ap} contain an additive constant $D_\alpha(\lambda(P_0)),  E_\alpha(\lambda(P_0))$ respectively. Speaking of the periods of an admissible differential, we are implicitly quotienting out this additive constant.
\end{Rem}
Clearly, the coefficients $A, B, C, D, E, R, S$ in \eqref{ap} completely characterize $\Omega$. A horizontal section is by definition a moduli-dependent family of differentials of the form \eqref{ap}, such that all the coefficients $A, B, C, D, E, R, S$ are \emph{independent of the moduli}. As a consequence,  all singularities, jumps and $a$-periods disappear when we differentiate with respect to the moduli while keeping $\lambda$ constant. At the same time, a double pole at $P_i$ must appear, since
\begin{align}\label{doublepole}
		&\hspace{50pt}\Omega=(\Omega(P_i)+O(\tau))d\tau	&\qquad \text{near}\; P_i, \;\; \tau=(\lambda-u_i)^{1/2}\nonumber\\
&\Rightarrow\hspace{30pt}	 \partial_i\Omega= \left(\frac{\Omega(P_i)}{2\tau^2}+O(1)\right)d\tau
\end{align}
Here and in the following, we use the notation $\Omega(P_i)$ to indicate the $0$th order term in the expansion of $\Omega$ at the $i$-th ramification point in powers of the natural local parameter. 
\begin{Def} A horizontal differential is \emph{covariantly constant} if $\partial_i\Omega=0$ for all $i$, or, equivalently, it is of the form
$$\sum_{k\in\mathbb{Z}}(c_k+b_k\log\lambda)\lambda^kd\lambda$$
with $c_k, b_k$ independent of the moduli.\\
We denote by $H(\mathcal{M})$ the space of horizontal differentials modulo covariantly constant differentials.
\end{Def}
It is clear from \eqref{ap} that an admissible differential on a given spectral curve of $\mathcal{M}$ extends uniquely to an horizontal differential on the whole $\mathcal{M}$. Thus the coefficients $A, B, C, D, E, R, S$ can serve as coordinates on $H(\mathcal{M})$; they are however subject to a set of linear constraints, partly of geometrical nature (e.g. the sum of residues must be zero) and partly coming from the fact that we are quotienting out covariantly constant differentials. More precisely we have:
\begin{Lemma}\label{basis} The following differentials form a basis of $H(\mathcal{M})$:
\begin{itemize}
\item[1.] Abelian differentials $\Omega_{a, r}^{(-k)},\; r=1, \dots, m_a-1,\; a=1, \dots, K, \;k=0, 1, \dots$ of the second kind having only a pole at $z_a$ with principal part
$$
\Omega_{a, r}^{(-k)}=\left(\frac{1}{\tau^{k m_a+1+r}}+ O(1)\right)d\tau \qquad \emph{at}\;\; z_a,\;\; \tau=\lambda^{1/m_a}
$$
and zero $a$-periods.
\item[2.]Abelian differentials $\Omega_{b, s}^{(k)},\; s=1, \dots, n_b-1,\; b=1, \dots, L, \;k=0, 1, \dots$ of the second kind having only a pole at $p_b$ with principal part
$$
\Omega_{b, s}^{(k)}=\left(\frac{1}{\tau^{k n_b+1+s}}+ O(1)\right)d\tau \qquad \emph{at}\;\; p_b,\;\; \tau=\lambda^{-1/n_b}
$$
and zero $a$-periods.
\item[3.] Abelian differentials $\Omega_{a}^{(-k)}, \; a=2, \dots, K, \; k=1, 2, \dots$ of the second kind having only a pole at $z_a$ with principal part
$$
\Omega_a^{(-k)}= d(\lambda^{-k})+\emph{reg}\qquad\emph{at}\;\; z_a
$$
and zero $a$-periods.
\item[4.] Abelian differentials $\Omega_{b}^{(k)}, \; b=2, \dots, L, \; k=1, 2, \dots$ of the second kind having only a pole at $p_b$ with principal part
$$
\Omega_b^{(k)}= d(\lambda^{k})+\emph{reg}\qquad\emph{at}\;\; p_b
$$
and zero $a$-periods.
\item[5.] Abelian differentials $\Phi^{(0)}_a, \; a=2, \dots, K$ of the third kind having simple poles at $z_1, z_a$ with residues $\pm 1$ respectively and zero $a$-periods.
\item[6.] Abelian differentials $\Psi^{(0)}_b, \; b=2, \dots, L$ of the third kind having simple poles at $p_1, p_b$ with residues $\pm 1$ respectively and zero $a$-periods.
\item[7.] Multivalued differentials $\Phi_a^{(k)}, \; a=2, \dots, K, \; k\in\mathbb{Z}\setminus 0$ with singularities only at $z_1, z_a$ of the form
\begin{align*}
\Phi_a^{(k)}	&=\frac{1}{m_1}d(\lambda^k\log\lambda)+\emph{reg}\qquad \emph{at}\;\;z_1\\
		&=\frac{1}{m_a}d(\lambda^k\log\lambda)+\emph{reg}\qquad \emph{at}\;\;z_a
\end{align*}
and zero $a$-periods.
\item[8.] Multivalued differentials $\Psi_b^{(k)}, \; b=2, \dots, K$ with singularities only at $p_1, p_b$ of the form
\begin{align*}
\Psi_b^{(k)}	&=\frac{1}{n_1}d(\lambda^k\log\lambda)+\emph{reg}\qquad \emph{at}\;\;p_1\\
		&=\frac{1}{n_b}d(\lambda^k\log\lambda)+\emph{reg}\qquad \emph{at}\;\;p_b
\end{align*}
and zero $a$-periods.
\item[9.] Multivalued differentials $\sigma_\alpha^{(k)}, \; \alpha=1, \dots, g, \; k\in\mathbb{Z}$ with jumps along the $a$-cycles of the form
$$
\sigma^{(k)}_\alpha(P+a_\alpha)-\sigma(P)=\lambda^{k-1}d\lambda
$$
and zero $a$-periods.
\item[10.] Multivalued differentials $\rho_\alpha^{(k)}, \; \alpha=1, \dots, g, \; k\in\mathbb{Z}$ with jumps along the $b$-cycles of the form
$$
\rho^{(k)}_\alpha(P+b_\alpha)-\rho(P)=\lambda^{k-1}d\lambda
$$
and zero $a$-periods.
\item[11.] Holomorphic differentials $\omega_\alpha, \;\alpha=1, \dots, g$ normalized on the $a$-periods:
$$
\oint_{a_\beta}\omega_\alpha=\delta_{\alpha\beta}
$$
\end{itemize}
\end{Lemma}
The main tool to associate flat geometries to horizontal differentials is a bilinear pairing on $H(\mathcal{M})$:
\begin{Lemma} Let
$$
\langle \cdot, \cdot\rangle: H(\mathcal{M})\times H(\mathcal{M})\to Funct(\mathcal{M})
$$
be the bilinear pairing defined by the formula
\begin{align}
\\
\langle\Omega_1, \Omega_2\rangle &=\sum_{a=1}^K\left[\sum_{k\leq-1}\frac{1}{k}B_{a,k-1}^{(1)}\emph{Res}_{z_a}\Big(\lambda^{k/m_a}\Omega_2\Big)- \sum_{l\in\mathbb{Z}}R_{a,l}^{(1)}\emph{p.v.}\int_{P_0}^{z_a}\lambda^l\Omega_2\right]+\nonumber\\
							&+\sum_{b=1}^L\left[\sum_{k\leq-1}\frac{1}{k}C_{b,k-1}^{(1)}\emph{Res}_{p_b}\Big(\lambda^{-k/n_b}\Omega_2\Big)-\sum_{l\in\mathbb{Z}}S_{b,l}^{(1)}\emph{p.v.}\int_{P_0}^{p_b}\lambda^l \Omega_2\right]+\nonumber\\
							&+\frac{1}{2\pi i}\sum_{\alpha=1}^g\left[\oint_{b_\alpha}D_\alpha^{(1)}(\lambda)\Omega_2-\oint_{a_\alpha}E_\alpha^{(1)}(\lambda)\Omega_2+A_\alpha^{(1)}\oint_{b_\alpha}\Omega_2\right]\nonumber
\\\nonumber
\end{align}
Here all the coefficients appearing refer to the first differential $\Omega_1$. The principal values are defined by subtraction of the divergent part of the integrals in the canonical local parameters and $P_0$ is an arbitrary horizontal section $\lambda(P_0)=\text{const}$.\\
Then the identity
\begin{equation}\label{resfor}
\frac{\partial}{\partial u_i}\langle \Omega_1, \Omega_2\rangle = \frac{\Omega_1(P_i)\Omega_2(P_i)}{2}
\end{equation}
holds true.
\end{Lemma}
\begin{proof} The proof of \cite{Du-Diff} works (with obvious modifications) in our more general setting.
\end{proof}
At this stage we select an horizontal differential $\phi\in H(\mathcal{M})$, called the \emph{quasi-momentum differential}. Setting
\begin{equation}\label{etai}
\eta_i\doteq\frac{\partial}{\partial u_i}\langle\phi, \phi\rangle=\frac{\phi(P_i)^2}{2}
\end{equation}
we associate to $\phi$ three metrics, all diagonal in the canonical coordinates:
\begin{equation}\label{1stm}
\eta\doteq\sum_{i=1}^n\eta_idu_i^2
\end{equation}
\begin{equation}\label{2ndm}
g\doteq\sum_{i=1}^n\frac{\eta_i}{u_i}du_i^2
\end{equation}
\begin{equation}\label{3rdm}
\tilde{\eta}\doteq\sum_{i=1}^n\frac{\eta_i}{u_i^2}du_i^2
\end{equation}
They are defined not globally on $\mathcal{M}$, but only on the dense open subset $\mathcal{M}_\phi=\{\phi(P_i)\neq0, i=1, \dots, n\}$. We will implicitly assume this restriction in the following.\\
\begin{Thm}\label{flatgeom} The metric $g$ is flat. Moreover
\begin{itemize}
\item If $\mu=(1, \dots, 1)$ then $\eta$ is flat.
\item If $\nu=(1, \dots, 1)$ then $\tilde{\eta}$ is flat.
\end{itemize}
Whenever two of the metrics $\eta, g, \tilde{\eta}$ are flat, they form a flat pencil.
\end{Thm}
We begin by noting that the curvature of the metrics \eqref{1stm}, \eqref{2ndm}, \eqref{3rdm} does not depend on the choice of the quasi-momentum differential:
\begin{Lemma} For all $\Omega_1, \Omega_2 \in H(\mathcal{M})$, the identity
\begin{equation}\label{gammaeq}
\Omega_1(P_i)\frac{\partial}{\partial u_i}\Omega_2(P_j)=\Omega_2(P_i)\frac{\partial}{\partial u_i}\Omega_1(P_j)\qquad \text{for all}\;\; i\neq j
\end{equation}
holds true. In particular, the rotation coefficients
$$
\gamma_{ij}=\frac{\partial_i\sqrt{\eta_j}}{\sqrt{\eta_i}}
$$
do not depend on the choice of $\phi$.
\end{Lemma}
\begin{proof}
Again, the proof of \cite{Du-Diff} applies almost verbatim to the case of double Hurwitz spaces.
\end{proof}
We will prove Theorem \ref{flatgeom} by explicitly constructing the flat coordinates of the three metrics. They result from the action of the infinitesimal automorphisms of $\mathbb{P}^1$ on the Hurwitz space: if $\mu, \nu$ are both trivial (i.e. equal to $(1, \dots, 1)$), any infinitesimal automorphism of $\mathbb{P}^1$ preserves the type of the Hurwitz cover, and we have an action of $sl(2, \mathbb{C})$ on $\mathcal{M}$. On the other hand, if the prescribed tangency condition at one point is nontrivial, only the subalgebra fixing that point acts on $\mathcal{M}$.\\
Introduce the following $sl(2, \mathbb{C})$ algebra of vector fields on the universal curve:
\begin{equation}\label{DeEe}
D_e\doteq \frac{\partial}{\partial\lambda}+\partial_e,\qquad \partial_e=\sum_{i=1}^n\frac{\partial}{\partial u_i}
\end{equation}
$$
D_E\doteq \lambda\frac{\partial}{\partial\lambda}+\partial_E,\qquad \partial_E=\sum_{i=1}^nu_i\frac{\partial}{\partial u_i}
$$
$$
D_{\tilde{e}}\doteq\lambda^2\frac{\partial}{\partial\lambda}+\partial_{\tilde{e}} ,\qquad \partial_{\tilde{e}}=\sum_{i=1}^nu_i^2\frac{\partial}{\partial u_i}
$$
As usual we let the vector fields on the universal curve act on horizontal differentials as the Lie derivative.
\begin{Prop}\label{HT0T}
\begin{itemize}\item[]
\item[(i)] For all $\Omega\in H(\mathcal{M})$, $D_E\Omega\in H (\mathcal{M})$ and we have the formula
\begin{equation}\label{LeibnitzE}
\partial_E\langle \Omega_1, \Omega_2\rangle=\langle D_E\Omega_1, \Omega_2\rangle+\langle \Omega_1, D_E\Omega_2\rangle+\emph{const}
\end{equation}
Moreover, $D_E$ has a $n$-dimensional kernel $H_0(\mathcal{M})=\{\Omega\in H(\mathcal{M}); D_E\Omega=0\}$, which is spanned by
\begin{itemize}
\item $\Phi^{(0)}_a\; a=2, \dots, K$.
\item $\Psi^{(0)}_b\; b=2, \dots, L$.
\item $\rho_\alpha^{(0)},\;, \alpha=1, \dots, g$.
\item $\omega_\alpha, \; \alpha=1, \dots, g$.
\end{itemize}
\item[(ii)] Let $\mu$ be trivial. Then for all $\Omega\in H(\mathcal{M})$, $D_e\Omega\in H(\mathcal{M})$ and we have the formula
\begin{equation}\label{Leibnitze}
\partial_e\langle \Omega_1, \Omega_2\rangle=\langle D_e\Omega_1, \Omega_2\rangle+\langle \Omega_1, D_e\Omega_2\rangle+\emph{const}
\end{equation}
Moreover, $D_e$ has a $n$-dimensional kernel  $H_T(\mathcal{M})=\{\Omega\in H(\mathcal{M}); D_e\Omega=0\}$, which is spanned by
\begin{itemize}
\item $\Omega^{(0)}_{b, s}, \; s=1, \dots, n_b-1, \;b=1, \dots, L$.
\item $\Omega^{(1)}_b,\; b=2, \dots, L$.
\item $\Psi^{(0)}_b\; b=2, \dots, L$.
\item $\rho_\alpha^{(1)}, \;\alpha=1, \dots, g$.
\item $\omega_\alpha, \; \alpha=1, \dots, g$.
\end{itemize}
\item[(iii)] Let $\nu$ be trivial. Then for all $\Omega\in H(\mathcal{M})$, $D_{\tilde{e}}\Omega\in H(\mathcal{M})$ and we have the formula
\begin{equation}\label{Leibnitztildee}
\partial_{\tilde{e}}\langle \Omega_1, \Omega_2\rangle=\langle D_{\tilde{e}}\Omega_1, \Omega_2\rangle+\langle \Omega_1, D_{\tilde{e}}\Omega_2\rangle+\emph{const}
\end{equation}
Moreover, $D_{\tilde{e}}$ has a $n$-dimensional kernel  $H_{\tilde{T}}(\mathcal{M})=\{\Omega\in H(\mathcal{M}); D_{\tilde{e}}\Omega=0\}$, which is spanned by
\begin{itemize}
\item $\Omega^{(0)}_{a, r}, \; r=1, \dots, m_a-1, \;a=1, \dots, K$.
\item $\Omega^{(-1)}_a,\; a=2, \dots, K$.
\item $\Phi^{(0)}_a\; a=2, \dots, K$.
\item $\rho_\alpha^{(-1)}, \;\alpha=1, \dots, g$.
\item $\omega_\alpha, \; \alpha=1, \dots, g$.
\end{itemize}
\end{itemize}
\end{Prop}
\begin{proof}Consider a horizontal differential $\Omega$ near a marked point, say a zero $z_a$ of $\lambda$. For simplicity let us assume that it is holomorphic there (the general case is treated in the same way):
$$
\Omega =(\Omega_0+O(\tau))d\tau\qquad \tau=\lambda^{1/m_a}
$$
where $\Omega_0$ is in general a function of $u_1, \dots, u_n$. It il clear that
$$
D_E\Omega=\left(\partial_E\Omega_0+\frac{1}{m_a}\Omega_0+O(\tau)\right)d\tau
$$
is still holomorphic at $z_a$. On the other hand
$$
D_e\Omega=\left(\frac{1-m_a}{m_a}\tau^{-m_a}\Omega_0+O(\tau^{1-m_a})\right)d\tau
$$
has a pole at $z_a$ unless $m_a=1$. Moreover, the principal part will as a rule depend on the moduli, which means that the obtained differential is not horizontal.\\
Repeating the same argument at a pole of $\lambda$ shows that $D_E$ automatically preserves horizontality, while $D_{\tilde{e}}$ only does if $n_b=1$ for all $b$.\\
The ``Leibnitz-type" formulas are easy consequences of \eqref{resfor}.\\
Finally the statements about the kernels of $D_e, D_E, D_{\tilde{e}}$ are checked immediately by looking at the action of the three operators on the basis elements of Lemma \ref{basis}.
\end{proof}
\begin{proof}[Proof of Theorem \ref{flatgeom}]
Let us prove that $g$ is flat by constructing explicitly flat coordinates. Let $\psi_1, \dots, \psi_n$  be a basis of $H_0(\mathcal{M})$ and $f_k\doteq \langle \psi_k, \phi\rangle$. Then, using \eqref{resfor} and Proposition 9, 
\begin{align*}
g^{-1}(df_k, df_l)	&= \sum_{i=1}^n g_i^{-1}(u)
\partial_if_k\partial_if_l=\\
					&=\sum_{i=0}^n\Big(\frac{2u_i}{\phi(P_i)^2}\Big)\Big(\frac{\phi(P_i)\psi_k(P_i)}{2}\Big)\Big(\frac{\phi(P_i)\psi_l(P_i)}{2}\Big)=\\
					&=\sum_{i=0}^nu_i\frac{\psi_k(P_i)\psi_l(P_i)}{2}=\partial_E\langle\psi_k, \psi_l\rangle= \text{const}					
\end{align*}
since $D_E\psi_k=D_E\psi_l=0$.\\
Note that the $f_k$ are independent functions. Indeed, assume $\partial_i\langle\psi, \phi\rangle=0$ for all $i=1, \dots, n$ for some $\psi\in H_0(\mathcal{M})$. Because of \eqref{resfor} and the fact that $\phi(P_i)\neq 0$ on $\mathcal{M}$ by assumption, this would imply $\psi(P_i)=0$ and therefore $\partial_i\psi=0$ for all $i$ (see \eqref{doublepole}). Thus $\psi$ would be covariantly constant and therefore zero in $H(\mathcal{M})$. This proves that the $f_k$ provide a set of flat coordinates for $g$.\\
Flatness of $\eta$ and $\tilde{\eta}$ are proved in the same way: the flat coordinates are given again by $f=\langle\psi, \phi\rangle$ for $\psi\in H_T(\mathcal{M}), H_{\tilde{T}}(\mathcal{M})$ respectively, whenever these spaces are defined.\\
For the compatibility statement, recall that for a diagonal metric
$$
g=\sum_{i=1}^ng_i du_i^2, 
$$
the flatness property is expressed in terms of the rotation coefficients in the form
$$
\partial_i\gamma_{jk}=\gamma_{ji}\gamma_{ik}\qquad i, j, k, \text{distinct}
$$
$$
\partial_i\gamma_{ij}+\partial_j\gamma_{ji}+\sum_{k\neq i, j}\gamma_{ki}\gamma_{kj}=0
$$
Given a second diagonal metric of the form
$$
\tilde{g}=\sum_{i=1}^n\frac{g_i}{f_i}du_i^2, 
$$
where the $f_i=f_i(u_i)$ are functions of the single variable $u_i$, the first part of the above system is automatically satisfied for $\tilde{g}$. The second part reads
$$
\frac{1}{2}f_i'\gamma_{ij}+\frac{1}{2}f_j'\gamma_{ji}+f_i\partial_i\gamma_{ij}+f_j\partial_j\gamma_{ji}+\sum_{k\neq i, j}f_k\gamma_{ki}\gamma_{kj}=0
$$
which is linear in the $f_i$. As a consequence any two solutions automatically give rise to a flat pencil.
\end{proof}
As a by-product of the proof we obtained:
\begin{Cor}\label{flatcoords}
\begin{itemize}\item[]
\item[(i)] A basis of flat coordinates of $g$ is given by
\begin{equation}\label{Hperiods}
p_a\doteq\langle\Omega_a, \phi\rangle
\end{equation}
where $\Omega_a$ runs through a basis of $H_0(\mathcal{M})$.
\item[(ii)] Let $\mu$ be trivial. A basis of flat coordinates of $\eta$ is given by
\begin{equation}\label{Htalpha}
t_\alpha\doteq\langle\Omega_\alpha, \phi\rangle
\end{equation}
where $\Omega_\alpha$ runs through a basis of $H_T(\mathcal{M})$.
\item[(iii)] Let $\nu$ be trivial. A basis of flat coordinates of $\tilde{\eta}$ is given by
\begin{equation}
s_\alpha\doteq\langle\Omega_\alpha, \phi\rangle
\end{equation}
where $\Omega_\alpha$ runs through a basis of $H_{\tilde{T}}(\mathcal{M})$.
\end{itemize}
\end{Cor}
\section{Frobenius structures and Whitham-type hierarchies}
Theorem \ref{flatgeom} naturally leads to the following generalization of Dubrovin's construction of Frobenius structure on (single) Hurwitz spaces \cite{Du-Diff,Du-2dtft}:
\begin{Thm}\label{Frobst} Let the quasi-momentum $\phi$ be homogeneous, in the sense that
$$
D_E\phi=\frac{1-d}{2}\phi
$$
for some constant $d$.
\begin{itemize}
\item[(i)] The metric $g$, together with the multiplication rule of vector fields
\begin{equation}\label{2ndssprod}
\frac{\partial}{\partial u_i}\star\frac{\partial}{\partial u_j}=\delta_{ij}\frac{1}{u_i}\frac{\partial}{\partial u_i}
\end{equation}
endows $\mathcal{M}$ with a semisimple dual-type Frobenius structure of charge $d$.
\item[(ii)] Let $\mu$ be trivial and $\phi\in H_T(\mathcal{M})$. Then the metric $\eta$ and the product
$$
\frac{\partial}{\partial u_i}\cdot\frac{\partial}{\partial u_j}=\delta_{ij}\frac{\partial}{\partial u_i}
$$
define a semisimple conformal Frobenius structure of charge $d$ on $\mathcal{M}$, and the dual-type structure (i) coincides with the almost-dual one.
\item[(iii)] If $\mu$ and $\nu$ are both trivial, then the conformal structure of (ii) is tri-hamiltonian.
\end{itemize}
\end{Thm}
\begin{proof}
We already proved that the intersecton form $g$ is flat, and the product \eqref{2ndssprod} is obviously associative and compatible with $g$. It remains to check that the flat coordinates of $g$ are homogeneous functions of degree $(1-d)/2$ of the canonical coordinates. This follows immediately from the explicit formula \eqref{Hperiods} and \eqref{LeibnitzE}.\\
Part (ii) is just a restatement of Dubrovin's construction.\\
To prove (iii), let us consider the homogeneity degrees of the flat coordinates \eqref{Htalpha} of $\eta$. Using Proposition 9 and the fact that $n_b=1$ by assumption, we find
\begin{align*}
&\partial_E \langle\phi, \Omega_\alpha\rangle=\frac{1-d}{2}\langle\phi, \Omega_\alpha\rangle \hspace{70pt}\text{for}\;\; \Omega_\alpha= \Psi^{(0)}, \;\omega\\
&\partial_E\langle\phi, \Omega_\alpha\rangle=\frac{3-d}{2}\langle\phi, \Omega_\alpha\rangle \hspace{70pt}\text{for}\;\; \Omega_\alpha=\Omega^{(1)}, \rho^{(1)}
\end{align*}
up to constants. Thus the flat coordinates split into two blocks of the same degree. This completes the proof.
\end{proof}
\begin{Rem} As we mentioned, the case (ii) recovers the standard construction, with the slight difference that the discriminant is automatically excluded in our setting (the canonical coordinates are non-zero by construction).  It is well-known (\cite{Du-2dtft}, Appendix G) that each irreducible component of the discriminant is an hyperplane in the flat coordinates of the intersection form. The above result can be seen as a generalization of this fact to the sub-strata of the discriminant obtained by further degenerating the ramification over $0$.\\
The whole picture may also be interpreted from the point of view of the Frobenius submanifold theory developed by Strachan \cite{St1,St2}: part (iii) of the theorem says that in the big moduli space of ramified coverings of the projective line of fixed genus and degree, there is an open and dense subset possessing a tri-hamiltonian Frobenius structure. The various double Hurwitz space can be viewed as \emph{natural submanifolds} obtained as the intersection of certain caustics and discriminants (clearly, the roles of $0$ and $\infty$ are interchangeable). While the restriction to these strata of the two metrics $\eta, \tilde{\eta}$ is in general curved, the restriction of $g$ is always flat. This leads to the dual-type Frobenius structure on each of the strata.
\end{Rem}
\vspace{5pt}
\begin{Rem} By a standard argument, the defintion of the metrics $\eta, g, \tilde{\eta}$ and of the respective products $\cdot, \; \star, \; \ast$ (see \eqref{ast}) can be recast in the form of Ginzburg-Landau-type formulas. Introduce the Abelian integral
$$
p(P)\doteq\text{p.v.}\int_{P_0}^P\phi
$$
where the principal value in defined by subtraction of the divergent part at the marked points in the natural local parameters. Then the map $\lambda$ can locally be represented as a function of the variable $p$, depending on the moduli as parameters; the local function $\lambda(p)$ is called the \emph{superpotential} of the Hurwitz space. We have:
\begin{align}\label{super1st}
\eta(X, Y)=&-\sum_{P=z_a, p_b}\text{Res}_P\frac{X(\lambda(p))Y(\lambda(p))}{d\lambda(p)}\phi^2\\
\eta(X, Y\cdot Z)=&-\sum_{P=z_a, p_b}\text{Res}_P\frac{X(\lambda(p))Y(\lambda(p))Z(\lambda(p))}{d\lambda(p)}\phi^2\nonumber
\end{align}
\begin{align}\label{super2nd}
g(X, Y)=&-\sum_{P=z_a, p_b}\text{Res}_P\frac{X(\log\lambda(p))Y(\log\lambda(p))}{d\log\lambda(p)}\phi^2\\
g(X, Y\star Z)=&-\sum_{P=z_a, p_b}\text{Res}_P\frac{X(\log\lambda(p))Y(\log\lambda(p))Z(\log\lambda(p))}{d\log\lambda(p)}\phi^2\nonumber
\end{align}
\begin{align}\label{super3rd}
\tilde{\eta}(X, Y)=&\sum_{P=z_a, p_b}\text{Res}_P\frac{X(1/\lambda(p))Y(1/\lambda(p))}{d(1/\lambda(p))}\phi^2\\
\tilde{\eta}(X, Y\ast Z)=&-\sum_{P=z_a, p_b}\text{Res}_P\frac{X(1/\lambda(p))Y(1/\lambda(p))Z(1/\lambda(p))}{d(1/\lambda(p))}\phi^2\nonumber
\end{align}
Here $X, Y, Z$ are arbitrary vector fields on the Hurwitz space, and their action on the superpotential is now to be understood at $p=\text{const}$ (see \cite{Du-2dtft} for details on the derivation).
\end{Rem}
We now proceed to the construction of a complete hydrodynamic hierarchy associated to the dual-type Frobenius structure on $\mathcal{M}$. It is based on the following
\begin{Prop} Let $\Omega\in H(\mathcal{M})$ and $h\doteq\langle\phi, \Omega\rangle$. We have
\begin{equation}\label{Hess}
\hat{\nabla}_i\hat{\nabla}_jh=\delta_{ij}\frac{1}{u_i}\frac{\partial}{\partial u_i}\langle\phi, D_E\Omega\rangle
\end{equation}
where $\hat{\nabla}$ is the Levi-Civita connection of the metric $g$.
\end{Prop}
\begin{proof}
For the off-diagonal elements of the covariant Hessian, we have to check that
\begin{equation}\label{diag}
\partial_i\partial_jh-\Gamma_{ij}^i\partial_ih-\Gamma_{ij}^j\partial_jh=0\qquad i\neq j
\end{equation}
where $\Gamma$ are the Christoffell symbols of the intersection form in the canonical coordinates (there is no summation over repeated indices in this formula). Using \eqref{resfor}, we compute
$$
\partial_i\partial_j h=\partial_i\left(\frac{\phi(P_j)\Omega(P_j)}{2}\right)=\frac{\partial_i\phi(P_j)\Omega(P_j)}{2}+\frac{\phi(P_j)\partial_i\Omega(P_j)}{2}
$$
$$
\Gamma_{ij}^i\partial_ih=\partial_j \log(g_i^{1/2})\partial_ih=\frac{\partial_j\phi(P_i)\Omega(P_i)}{2}
$$
$$
\Gamma_{ij}^j\partial_jh=\partial_i \log(g_j^{1/2})\partial_jh=\frac{\partial_i\phi(P_j)\Omega(P_j)}{2}
$$
Then diagonality of the Hessian follows from \eqref{gammaeq}.\\
For the diagonal elements, we compute
$$
\hat{\nabla}_i\hat{\nabla}_ih=\partial_i^2h-\sum_{j=1}^n\Gamma_{ii}^j\partial_jh
$$
$$
\Gamma_{ii}^j=-\frac{u_j}{u_i}\partial_i\log\phi(P_j),\;\; i\neq j\qquad\qquad\Gamma_{ii}^i=\partial_i\log\phi(P_i)-\frac{1}{2u_i}
$$
which yields
\begin{align*}
\hat{\nabla}_i\hat{\nabla}_ih	&= \partial_i\left(\frac{\phi(P_i)\Omega(P_i)}{2}\right)-\frac{\partial_i\phi(P_i)\Omega(P_i)}{2}+\\
			&+\frac{1}{2u_i}\frac{\phi(P_i)\Omega(P_i)}{2}+\sum_{j\neq i}\frac{u_j}{u_i}\frac{\partial_i\phi(P_j)\Omega(P_j)}{2}=\\
						&=\frac{1}{2}\left(\phi(P_i)\partial_i\Omega(P_i)+\frac{1}{2u_i}\phi(P_i)\Omega(P_i)+\sum_{j\neq i}\frac{u_j}{u_i}\phi(P_i)\partial_j\Omega(P_i)\right)=\\
						&=\frac{1}{u_i}\frac{\phi(P_i)(\frac{1}{2}\Omega(P_i)+\partial_E\Omega(P_i))}{2} = \frac{1}{u_i}\partial_i\langle\phi, D_E\Omega\rangle
\end{align*}
\end{proof}
Recall that the flat metric $g$ induces a Poisson bracket of hydrodynamic type on the formal loop space $\mathcal{LM}=\{X\in S^1\mapsto p(X)\in\mathcal{M}\}$:
\begin{equation}\label{2ndPoiss}
\{p^a(X), p^b(Y)\}=g^{ab}\delta'(X-Y)
\end{equation}
where $p^a$ are flat coordinates of $g$.
\begin{Cor}
\begin{itemize}
\item[]
\item[(i)] Let $I(\mathcal{M})\doteq\{h=\langle\phi, \Omega\rangle; \;\Omega\in H(\mathcal{M})\}$. Consider a local functional of hydrodynamic type with density with density in $I(\mathcal{M})$: 
\begin{equation}\label{locfun}
H[p]\doteq\int h(p)dX\qquad h=\langle\phi, \Omega\rangle\in I(\mathcal{M})
\end{equation}
Then its hamiltonian flow with respect to \eqref{2ndPoiss} can be written in the ``Flaschka-Forest-McLaughlin form" \cite{FFM}
\begin{equation}\label{FFM}
\partial_T\phi=\partial_X(D_E\Omega)
\end{equation}
\item[(ii)] Any two functionals of the form \eqref{locfun} are in involution with respect to \eqref{2ndPoiss}. Moreover, $I(\mathcal{M})$ is a complete family of conservation laws in the sense of Tsarev \cite{Ts}.
\item[(iii)] Assume $D_E\phi=0$ (i.e. d=1). Then for any pair of densities $h_i=\langle\phi, \Omega_i\rangle\in I(\mathcal{M}), \; i=1, 2$ we have
\begin{equation}
\partial_{T^1}\partial_Eh_2=\partial_X\langle D_E\Omega_1, D_E\Omega_2\rangle=\partial_{T^2}\partial_Eh_1
\end{equation}
where $\partial_{T^i}$ denotes the hamiltonian flow of the local functional with density $h_i$. In other words, $\partial_E$ acts as a tau-symmetry operator \cite{DuZh-Norm} in the space of conservation laws $I(\mathcal{M})$.
\item[(iv)] Let $H_k(\mathcal{M}), \; k\in\mathbb{Z}$ be the $n$-dimensional subspace of $H(\mathcal{M})$ consisting of homogeneous differentials of degree $k$,
\begin{equation}\label{Hk}
H_k(\mathcal{M})\doteq\{\Omega\in H(\mathcal{M});\;D_E\Omega=k\Omega\}
\end{equation}
Explicitly, for positive $k$ it is generated by
\begin{itemize}
\item $\Phi_a^{(k)}, \; a=2, \dots, K$.
\item $\Omega_b^{(k)}, \; b=2, \dots, L$.
\item $\sigma_\alpha^{(k)}, \; \alpha=1, \dots, g$.
\item $\rho_\alpha^{(k)}, \; \alpha=1, \dots, g$
\end{itemize}
and similarly for negative $k$, with $\Omega_a^{(-k)}$ instead of $\Phi_a^{(k)}$ and $\Psi_b^{(k)}$ instead of $\Omega_b^{(k)}$. Then, for all $\Omega\in H_k(\mathcal{M})$, $\langle\phi, \Omega\rangle$ is a flat function of the deformed connection
$$
\hat{\nabla}^{(z)}_X Y=\hat{\nabla}_X Y+z X\star Y
$$
for the value $z=k$ of the deformation parameter.
\end{itemize}
\end{Cor}
\begin{proof} $\;$\\
\indent (i) Take the leading term of the expansion of \eqref{FFM} at a ramification point $P_i$ and use the proposition:
\begin{align*}
\partial_Tu_i	&=\frac{(D_E\Omega)(P_i)}{\phi(P_i)}\partial_Xu_i=\frac{2u_i}{\phi(P_i)^2}\frac{\partial_i\langle\phi, D_E\Omega\rangle)}{2u_i}\partial_Xu_i=\\
			&=g_i^{-1}\hat{\nabla}_i\hat{\nabla}_i\langle\phi, \Omega\rangle\;\partial_Xu_i=\{H, u_i\}
\end{align*}
\indent (ii) The last computation shows that  the canonical coordinates are common Riemann invariants for the Hamiltonian flows of the functionals \eqref{locfun}, which immediately implies the involutivity statement. The completeness statement can be proven using the same method of \cite{Ts}: we have to show that $I(\mathcal{M})$ is dense in the space of solutions of the diagonality equation \eqref{diag}. The map
$$
h\;\;\to\;\; (\psi_1, \dots, \psi_n), \;\; \psi_i\doteq \phi(P_i)^{-1}\frac{\partial h}{\partial u_i}
$$
turns \eqref{diag} into the system
\begin{equation}\label{llinsys}
\partial_i\psi_j=\gamma_{ij}\psi_i\qquad i\neq j
\end{equation}
where $\gamma_{ij}$ are the rotation coefficients of the metric $\eta$. The solutions are locally parametrized by $n$ functions of one variable. Since the metric $\eta$ is not flat, the rotation coefficients are not translation invariant; instead we can use the homogeneity $\partial_E\gamma_{ij}=-\gamma_{ij}$ to complete \eqref{llinsys} to
\begin{equation}\label{llinsys2}
\left\{\begin{array}{l}
\partial_i\psi_j=\gamma_{ij}\psi_i\qquad i\neq j\\
\partial_E\psi_i = z\psi_i\\
\end{array}\right.
\end{equation}
where $z$ is a spectral parameter.
Taking this basis of $H_k(\mathcal{M})$ (see \eqref{Hk}) and the corresponding densities in $I(\mathcal{M})$, the above procedure will produce $n$ independent solutions of \eqref{llinsys2} for $z=k$. Doing this for each $k$ one constructs a full family of homogeneous, integer degree solutions of \eqref{llinsys}; restricting them to an integral curve of $\partial_E$ we obtain a basis of continuous vector-valued functions on the curve, which are suitable Cauchy data for \eqref{llinsys}. This proves completeness of $I(\mathcal{M})$.\\
\indent (iii) For $D_E\phi=0$, we have
\begin{align*}
\partial_{T^1}\partial_Eh_2	&=\partial_{T^1}\langle\phi, D_E\Omega_2\rangle=\sum_{i=1}^n\partial_i\langle \phi, D_E\Omega_2\rangle \partial_{T^1}u_i=\\
						&=\sum_{i=1}^n\frac{(D_E\Omega_1)(P_i)(D_E\Omega_2)(P_i)}{2}\partial_Xu_i=\sum_{i=1}^n\partial_i\langle D_E\Omega_1, D_E\Omega_2\rangle\partial_Xu_i=\\
						&=\partial_X\langle D_E\Omega_1, D_E\Omega_2\rangle
\end{align*}
The last expression is clearly symmetric under the exchange of 1 and 2.\\
\indent (iv) Rewriting \eqref{Hess} in the flat coordinates $p^a$ of the intersection form and using $\langle \phi, D_E\Omega\rangle=kh$ we obtain
\begin{equation}\label{twistedk}
\partial_a\partial_b h=k\hat{c}_{ab}^{d}\partial_dh
\end{equation}
where $\hat{c}_{ab}^c$ are the structure constants of $\,\star$.  The last equation is  just $\hat{\nabla}^{(z=k)}dh =0$.
\end{proof}
Let us fix for each $k$ a basis $\psi_{1, k}, \dots, \psi_{n, k}$ of $H_k(\mathcal{M})$. The hierarchy of quasi-linear PDEs
\begin{equation}\label{WhithamH}
\partial_{T^{a, k}}\phi=\partial_X\psi_{a, k}\qquad a=1, \dots, n, \;\; k\in\mathbb{Z}\setminus 0
\end{equation}
is called the \emph{Whitham-type hierarchy} associated to the double Hurwitz space $\mathcal{M}$. Each of the flows is hamiltonian with respect to \eqref{2ndPoiss}, with hamiltonian density $h_{a, k}\doteq\frac{1}{k}\langle \phi, \psi_{a, k}\rangle$.  Note that the level $k=0$ is excluded in \eqref{WhithamH} since the functions $\langle\phi, \psi_{a, 0}\rangle$ are densities of Casimirs of \eqref{2ndPoiss} (they are flat coordinates of $g$).
\begin{Rem} For $d=1$ equation \eqref{twistedk} may be rewritten as
$$
\partial_a\partial_b h=\hat{c}_{ab}^{d}\partial_d\partial_Eh
$$
which is the usual defining equation for the hamiltonian densities of the principal hierarchy \cite{Du-2dtft, DuZh-Norm} of a Frobenius manifold (recall that here $\partial_E$ plays the role of the unit vector field). There is however a structural difference here in the fact that the unit vector field acts as a diagonal operator ($\partial_E h_{a, k}=kh_{a, k}$) on the basis, as opposed to a recursion operator. This is due to the fact that in the construction of the principal hierarchy the hamiltonian densities are obtained as coefficients of the series expansion at $z=0$ of the deformed flat coordinates, while in \eqref{WhithamH} they are deformed flat coordinates evaluated at integer values of the deformation parameters. In the setting of double Hurwitz spaces (where the flat coordinates being deformed are morally those of the intersection form) the second way of defining the flows appears to be more natural: e.g. it admits a Flaschka-Forest-McLaughlin representation, while the first doesn't.\\
\end{Rem}
A rich family of solutions of \eqref{WhithamH} can be found via the generalized hodograph method: the solution is defined implicitly by the equations
\begin{equation}\label{hodograph}
\left.\left(\sum_{k\neq0 }\sum_{a=1}^n\tilde{T}^{a, k}\psi_{a, k}+X\phi\right)\right|_{P_i}=0\qquad i=1, \dots, n
\end{equation}
where $\tilde{T}^{a, k}\doteq T^{a, k}-C^{a, k}$ and $C^{a, k}$ are arbitrary constants which parametrize the family. For $d=1$, one easily checks that the tau-function of \eqref{hodograph} is given by
\begin{equation}\label{Whithamtau}
\log\tau(X, T)=\frac{1}{2}\sum \tilde{T}^{a, k}\tilde{T}^{b, l}\langle\psi_{a, k}, \psi_{b, l}\rangle+X\sum \tilde{T}^{a, k}\langle\phi,\psi_{a, k}\rangle+\frac{1}{2}X^2\langle\phi, \phi\rangle 
\end{equation}
\section{Examples}
\subsection{The dispersionless q-deformed Gelfand-Dickey hierarchy}
In \cite{Fr-Def, FrReSTS} the authors define a difference version of the Gelfand-Dickey hierarchy, based on a q-deformation of the classical $\mathcal{W}$-algebra of $sl(n)$. The Lax operator has the form
$$
L = \Lambda^{n+1}+\alpha_1\Lambda^{n}+\dots+\alpha_n\Lambda+\alpha_{n+1}
$$
where $\alpha_i=\alpha_i(l)$ are function of the discrete spatial variable $l\in\mathbb{Z}$ and $\Lambda$ is the difference operator\footnote{Equivalently, we may view the dependent variables $\alpha_i$ as functions of the ``slow" spatial variable $X=\epsilon l$ and $\Lambda$ as the shift operator $\Lambda\alpha(X)=\alpha(X+\epsilon)$. In the original papers the authors use the multiplicative version $\Lambda \alpha(z)=\alpha (q z)$; the additive and multiplicative formulations are identified by setting $z=e^X, \; q=e^\epsilon$.}
\begin{equation}\label{Lambda}
\Lambda \alpha(l)=\alpha(l+1)
\end{equation}
The equations of the hierarchy are given by the usual formulas
\begin{equation}\label{q-defGD}
\frac{\partial}{\partial t^p}L=[L, (L^{p/(n+1)})_+]\qquad p=1, 2, \dots, \;\;p\neq k(n+1)
\end{equation}
where $L^{1/(n+1)}$ is the unique pseudo-difference operator $\Lambda + \beta_0+\beta_{-1}\Lambda^{-1}+\dots$ such that $(L^{1/(n+1)})^{n+1}=L$ and the subscript $+$ denotes the polynomial part in $\Lambda$. It is proved in \cite{Fr-Def} that the equations \eqref{q-defGD} have bi-hamiltonian structure and commute pairwise.\\
A straightforward computation shows that \eqref{q-defGD} implies $\partial\alpha_{n+1}/\partial t^p=0$, so the hierarchy is naturally reduced to the submanifold $\alpha_{n+1}=\text{const}$; let us set it to zero. The reduced Lax operator is
\begin{equation}\label{redL}
L = \Lambda^{n+1}+\alpha_1\Lambda^{n}+\dots+\alpha_n\Lambda
\end{equation}
By definition, the dispersionless limit is obtained by setting $X=\epsilon l, T^p=\epsilon t^p$ and taking the $\epsilon\to 0$ limit in all the equations. An equivalent approach consists in taking the semi-classical limit of the Lax formalism: one substitutes $L$ with its symbol
\begin{equation}\label{FRsymbol}
L\to\lambda= z^{n+1}+\alpha_1 z^{n}+\dots+\alpha_n z
\end{equation}
and the commutator with the classical Poisson bracket
\begin{equation}\label{clPB}
[L, M]\to \{\lambda, \mu\}=z\frac{\partial\lambda}{\partial z}\frac{\partial\mu}{\partial X}-z\frac{\partial\mu}{\partial z}\frac{\partial\lambda}{\partial X}
\end{equation}
where $L, M$ are arbitrary difference operators and $\lambda, \mu$ the corresponding symbols. Here in short we substituted $\epsilon\partial_X\to p$ and put $z=e^p$.\\
This approach leads immediately to the relevant Hurwitz space: the symbol of the Lax operator is to be identified with the superpotential, and the quasi-momentum with the exact differential $dp$. Here the symbol \eqref{FRsymbol} corresponds to
$$
g=0,\; \nu=(1, \dots, 1),\; \mu=(n+1)
$$ and the quasi-momentum is the third kind differential
\begin{equation}\label{FRphi}
\phi\doteq dp\equiv\frac{dz}{z}
\end{equation}
with simple poles at the pole of $\lambda$ and at one of the zeros. The coefficients $\alpha_1, \dots, \alpha_n$ are global coordinates on the Hurwitz space. Here as usual we exclude the points where the polynomial is not sufficiently generic; specifically we assume that neither $\lambda$ nor $\lambda'$ have multiple roots.\\
The dispersionless q-difference Gelfand-Dickey hierarchy can be rewritten in the Flaschka-Forest-McLaughlin form
$$
\partial_{T^p}\phi=\partial_X\Omega^{(k)}_{1, s}\qquad p=k(n+1)+s
$$
\indent In the standard formulation of the Gelfand-Dickey hierarchy, the differential operator $\partial_x$ replaces $\Lambda$ in the Lax operator. The dispersionless limit is described by the same Hurwitz space, but with a different quasi-momentum, namely
$$
\phi=dz
$$
From the point of view of the underlying Frobenius structures (see Theorem \ref{Frobst}), this is a substantial difference: on one hand $dz$ belongs to $H_T(\mathcal{M})$ and thus leads to a conformal Frobenius structure on $\mathcal{M}$ with metric $\eta$ (the standard Frobenius structure on the orbit space of the reflection group $A_n$). On the other hand $dz/z$ does not belong to $H_T(\mathcal{M})$, so there is no conformal Frobenius structure associated to the first metric. What we do have is a dual-type structure of charge one (i.e. non-conformal) on the ``almost-dual" side, since $dz/z\in H_0(\mathcal{M})$. Note however that the metric $\eta$ is still flat (which is in agreement with the bihamiltonian nature of the q-difference version), as this property only depends on the Hurwitz space and not on the choice of the quasi-momentum. As a matter of fact, in the change of quasi-momentum from $dz$ to $dz/z$, the only axiom that is lost at the level of the first structure is flatness of the unit.\\
Let us illustrate these facts in the specific case $n=3$. The superpotential is 
$$
\lambda=z^4+\alpha_1z^3+\alpha_2z^2+\alpha_3z
$$
and the quasi-momentum is \eqref{FRphi}. We begin by considering the first structure; according to Corollary \ref{flatcoords} and Proposition \ref{HT0T}, a basis of flat coordinates of $\eta$ is given by
\begin{align*}
&\langle \phi, \Omega^{(0)}_{1, 1}\rangle=-\text{Res}_{z=\infty}\lambda^{1/4}\phi=\frac{1}{4}\alpha_1\\
&\langle \phi, \Omega^{(0)}_{1, 2}\rangle=-\frac{1}{2}\text{Res}_{z=\infty}\lambda^{1/2}\phi=\frac{1}{16}(4\alpha_2-\alpha_1^2)\\
&\langle \phi, \Omega^{(0)}_{1, 3}\rangle=-\frac{1}{3}\text{Res}_{z=\infty}\lambda^{3/4}\phi=\frac{1}{384}(5\alpha_1^3-24\alpha_1\alpha_2+96\alpha_3)
\end{align*}
Adjusting slightly the normalizations, we set
\begin{align*}
&t_1=-\frac{1}{4}\alpha_1\\
&t_2=\frac{1}{8}\alpha_1^2-\frac{1}{2}\alpha_2\\
&t_3=-\frac{5}{96}\alpha_1^3+\frac{1}{4}\alpha_1\alpha_2-\alpha_3
\end{align*}
The superpotential in terms of the $t_\alpha$ takes the form
$$
\lambda=z^4-4t_1z^3+\Big(4t_1^2-2t_2\Big)z^2-\left(\frac{2}{3}t_1^2-2t_2t^3+t_3\right)z
$$
Using \eqref{super1st} one easily computes the metric $\eta$ and the structure constants of $\cdot$ in the coordinates $t_\alpha$. We find
$$
\eta=\left(\begin{array}{ccc}
0 &0 &1\\
0 &1 &0\\
1 &0 &0\end{array}\right)
$$
and the prepotential
\begin{align*}
F(t_1, t_2, t_3)= -	&\left(\frac{1}{4}t_3^2t_2+\frac{1}{2}t_3t_2^2t_1+\frac{1}{3}t_3t_2t_1^3+\frac{1}{4}t_3^2t_1^2+\frac{1}{30}t_3t_1^5\right.\\
		&\left.+\frac{1}{12}t_2^4+\frac{1}{3}t_2^3t_1^2+\frac{1}{3}t_2^2t_1^4+\frac{1}{9}t_2t_1^6+\frac{1}{63}t_1^8\right)
\end{align*}
It is quasi-homogeneous of degree $8/3$ with respect to the Euler vector field
$$
\partial_E =t_3\frac{\partial}{\partial t_3}+\frac{2}{3}t_2\frac{\partial}{\partial t_2}+\frac{1}{3}t_1\frac{\partial}{\partial t_1}
$$
It does not satisfy the flatness of the unit axiom; indeed, the unit has the form
$$
\partial_e=\frac{1}{3t_3-6t_2t_1+2t_1^3}\left(6(t_2-t_1^2)\frac{\partial}{\partial t_3}+6t_1\frac{\partial}{\partial t_2}-3\frac{\partial}{\partial t_1}\right)
$$
The denominator is a factor in the discriminant of the polynomial $\lambda$, and hence does not vanish on $\mathcal{M}$ by assumption.\\
Let us now consider the second structure.  Again from Corollary \ref{flatcoords} and Proposition \ref{HT0T} we have that the functions
\begin{align*}
&\langle \phi, \Phi^{(0)}_2\rangle=\text{p.v.}\int_{z_1}^{z_2}\phi=\log z(z_2)-\log z(z_1)\\
&\langle \phi, \Phi^{(0)}_3\rangle=\text{p.v.}\int_{z_1}^{z_3}\phi=\log z(z_3)-\log z(z_1)\\
&\langle \phi, \Phi^{(0)}_4\rangle=\text{p.v.}\int_{z_1}^{z_4}\phi=\log z(z_4)-\log z(z_1)
\end{align*}
form a basis of flat coordinates of $g$. Since $z(z_1)=0$ by our choice of the uniform parameter $z$, the terms $\log z(z_1)$ only contribute an additive constant. Let us set
\begin{align*}
&p_1=4\log z(z_2)\\
&p_2= \log z(z_3)-\log z(z_2)\\
&p_3=\log z(z_4)-\log z(z_2)
\end{align*}
In terms of the coordinates $p_a$ the superpotential becomes
$$
\lambda= z(z-e^{p_1/4})(z-e^{p_1/4+p_2})(z-e^{p_1/4+p_3})
$$
Substituting in \eqref{super2nd} we compute the intersection form in the coordinates $p_a$, 
$$
g=-\left(\begin{array}{ccc}
3/4 &1 &1\\
1 &2 &1\\
1 &1 &2
\end{array}\right)
$$
and the almost-dual prepotential
\begin{align*}
\hat{F}(p_1, p_2, p_3)= 	-&\left(\frac{1}{8}p_1^3+\frac{1}{2}p_1^2p_2+\frac{1}{2}p_1^2p_3+p_1p_2^2+p_1p_2p_3\right.\\
				&\left.+p_1p_3^2+p_2^2p_3+\frac{2}{3}p_2^3+\frac{1}{2}p_2p_3^2+\frac{5}{6}p_3^3\right)\\
				&+\text{Li}_3(e^{p_2})+\text{Li}_3(e^{p_3})+\text{Li}_3(e^{p_2-p_3})
\end{align*}
where $\text{Li}_3$ stands for the tri-logarithmic function
\begin{equation}\label{trilog}
\text{Li}_3(x)\doteq\sum_{k\geq 1}\frac{x^k}{k^3}
\end{equation}
It is readily checked that
$$
\frac{\partial^3\hat{F}}{\partial p_1\partial p_a\partial p_b}=g_{ab}
$$
so that $\partial/\partial p_1$ is the unit of the algebra. Clearly, $\hat{F}$ is not quasi-homogeneous because of the presence of the tri-logarithmic part.

\subsection{The Whitham-averaged sine-Gordon equation}
The sine-Gordon equation
\begin{equation}\label{sineG}
(\partial_t^2-\partial_x^2)U+\sin U=0
\end{equation}
is a well-known example of integrable equation by the inverse scattering method. The construction of $g$-phase sine-Gordon solutions and the theory of their slow modulation (spectral theory, invariant representation and hamiltonian formalism) were developed in the eighties by Forest, McLaughlin and collaborators \cite{FoML-SinG1,FoML-SinG2,EFMM}. We recall here some of the main results from their work, referring to the original papers for details.\\
For all positive $g$, \eqref{sineG} possesses a family of exact $g$-phase solutions of the form
\begin{equation}\label{g-phaseSG}
U(x, t)=W_g(\theta_1(x, t), \dots, \theta_g(x, t); u_1, \dots, u_{2g})
\end{equation}
where the function $W_g$ is periodic in each of the arguments $\theta_i$, each $\theta_i$ depends linearly on $x$ and $t$, and the variables $u_1, \dots, u_{2g}$ are parametrize the family. More precisely, the solution $W_g$ is explicitly expressed in terms of the theta function on the Jacobian of the spectral curve
\begin{equation}\label{sineGcurve}
\mu^2=\lambda(\lambda-u_1)\cdots(\lambda-u_{2g})
\end{equation}
The wave numbers and frequencies (which describe the dependence of the phases on space and time) are also uniquely determined by the curve. It follows that the $g$-phase solutions of the sine-Gordon equation are parametrized by the Hurwitz space $\mathcal{M}$ of type
$$
g,\; \mu=(2),\; \nu=(2)
$$
which is by definition the moduli space of hyperelliptic curves \eqref{sineGcurve}. Physical solutions of sine-Gordon can be locally approximated by slow modulation of the exact solutions \eqref{g-phaseSG}, i.e. allowing a dependence $u_i=u_i(X, T)$ of the parameters on the slow space-time variables $X=\epsilon x, T=\epsilon t$; such dependence is described by the Whitham equations, which can be expressed compactly in the Flaschka-Forest-McLaughlin form
\begin{equation}\label{With}
\partial_T\Omega^{(-)}=\partial_X\Omega^{(+)}
\end{equation}
where $\Omega^{(\pm)}$ are the horizontal differentials on $\mathcal{M}$ with double poles at $\lambda=0, \lambda=\infty$ and principal parts
$$
\Omega^{(\pm)}=\left\{
\begin{array}{l}
\left(\frac{1}{\tau^2}+O(1)\right)d\tau\qquad\text{at}\;\lambda=\infty, \;\tau=\lambda^{-1/2}\\
\\
\left(\pm\frac{1}{16}\frac{1}{\tau^2}+O(1)\right)d\tau\qquad\text{at}\;\lambda=0,\;\tau=\lambda^{1/2}
\end{array}\right.
$$
Corollary 2 implies that \eqref{With} is Hamiltonian with respect to the hydrodynamic bracket \eqref{2ndPoiss}, with the quasi-momentum differential 
$$
\phi={\Omega^{(-)}}
$$
and the hamiltonian density
$$
h=-2\langle \Omega^{(-)}, \Omega^{(-)}\rangle
$$
(indeed, $D_E\Omega^{(-)}=-\Omega^{(+)}/2$).\\
\indent
Let us consider the case $g=1$ in detail. The Hurwitz space is the 2-dimensional moduli space of hyperelliptic curves
\begin{equation}
\mu^2=\lambda(\lambda-u_1)(\lambda-u_2)
\end{equation}
In the Weierstrass uniformization,
$$
\lambda=\wp(z;\omega_1, \omega_2)-e_3,\qquad\mu=\frac{1}{2}\wp'(z;\omega_1, \omega_2)
$$
$$
u_i = e_i-e_3, \qquad i=1,2
$$
where $\wp$ is the Weierstrass elliptic function with periods $\omega_1, \omega_2$ and $e_i = \wp(\omega_i/2)$, $i=1, 2, 3$ ($\omega_3 = \omega_1 + \omega_2$).\\
Because the differential $\Omega^{(-)}$ is not homogeneous, using it as the quasi-momentum only produces a flat metric $g$, but not a dual-type Frobenius structure on $\mathcal{M}$ (i.e. neither the flatness of the unit nor the quasi-homogeneity axiom are satisfied by the associated prepotential). We will consider instead the dual-type Frobenius structure associated to the normalized holomorphic differential
\begin{equation}\label{qmom}
\phi\doteq\frac{dz}{\omega_1}=\frac{d\lambda}{2\mu}
\end{equation}
The general Flaschka-Forest-McLaughlin flow has the form
$$
\partial_T\phi=\partial_X\Omega\qquad\Omega\in H(\mathcal{M})
$$
The original equation \eqref{With} is recovered from the pair of flows
$$
\partial_{T^{(-)}}\phi=\partial_X\Omega^{(-)}\qquad \partial_{T^{(+)}}\phi=\partial_X\Omega^{(+)}
$$
eliminating $X$ and interpreting $T^{(+)}$ as the new spatial variable.\\
Let us compute the almost-dual prepotential $\hat{F}$ associated to the choice \eqref{qmom}.\\
The intersection form reads
\begin{equation}\label{gsg}
g=\frac{\phi(\omega_1/2)^2}{2u_1}(du_1)^2+\frac{\phi(\omega_2/2)^2}{2u_2}du_2^2=\frac{1}{2\omega_1^2(u_1-u_2)}\left(\frac{du_1^2}{u_1^2}-\frac{du_2^2}{u_2^2}\right)
\end{equation}
A basis of $H_0(\mathcal{M})$ is provided by $\phi$ itself plus the multivalued holomorphic differential $\rho_1^{(0)}$, with the jump
$$
\rho_1^{(0)}(P+\omega_2)-\rho_1^{(0)}(P)=\frac{d\lambda}{\lambda}
$$
and vanishing $a$-period. Therefore a basis of flat coordinates for $g$ is obtained from the pairings
$$
\langle\phi, \phi\rangle \propto \oint_b\frac{dz}{\omega_1},\qquad\langle \phi, \rho_1^{(0)}\rangle\propto\oint_a\log\lambda\frac{dz}{\omega_1}
$$
The first integral is just the modular parameter $\omega_2/\omega_1$. To compute the second, we introduce the Abelian integral\footnote{Here and in the rest of this section $\eta_i=\zeta(\omega_i/2)$, where $\zeta$ is the Weierstrass zeta function.}
$$
q(z)\doteq -\zeta(z) +\frac{2\eta_1}{\omega_1}z
$$
with periods
$$
q(z+\omega_1)-q(z)=0\qquad q(z+\omega_2)-q(z)=\frac{2\pi i}{\omega_1}
$$
Let $\tilde{C}$ be obtained by cutting the elliptic curve along the $a$ and $b$ cycles plus an additional path joining $\lambda=\infty$ ($z=0$) and $\lambda=0$ ($z=\omega_3/2$). Let $\gamma$ be a small loop around this additional cut, and set $\psi\doteq\log\lambda dz/\omega_1$. Then $q(z)\psi(z)$ is a holomorphic single-valued differential on $\tilde{C}$, so that
 \begin{align*}
 0	&=\int_{\partial\tilde{C}}q\psi=\left(\oint_a-\oint_{a'}\right)q\psi+\left(\oint_b-\oint_{b'}\right)q\psi- \oint_\gamma q\psi=\\
 	&=\oint_{z\in a}\psi(z)\Big(q(z)-q(z+\omega_2)\Big) +\oint_{z\in b}\psi(z)\Big(q(z)-q(z-\omega_1)\Big)-\oint_\gamma q\psi
 \end{align*}
 from which
 \begin{equation}
 \oint_a\log\lambda\frac{dz}{\omega_1}=-\frac{\omega_1}{2\pi i}\oint_\gamma q\psi
 \end{equation}
 The right hand side equals
 \begin{equation}\label{pvres}
-\frac{\omega_1}{2\pi i} \oint_\gamma q\psi=\;\text{p.v.}\int_{z=0}^{z=\omega_3/2}q(z)dz+\pi i\;\text{Res}_{z=0}\;q(z)dz
 \end{equation}
where as usual the meaning of the principal value is to subtract the divergent part of the integral in the canonical local parameters at the end-points. The residue term in the right hand side of \label{pvres} is just a constant since $\text{Res}_{z=0}\;\zeta(z)dz=1$. The principal value is
$$
\text{p.v.}\int_{0}^{\omega_3/2}\left(-\zeta(z)+\frac{2\eta_1}{\omega_1}z\right)dz=\text{p.v.}\left[-\log \sigma(z)+\frac{\eta_1}{\omega_1}z^2\right]_0^{\omega_3/2}=
$$
$$
=-\log\sigma(\omega_3/2)+\text{p.v.}\log\sigma(z)\vert_{z=0}+\frac{\eta_1}{4\omega_1}(\omega_1+\omega_2)^2
$$
in terms of the Weierstrass $\sigma$-function. We have
\begin{itemize}
\item $\text{p.v.}\log\sigma(z)\vert_{z=0}=0$, since
\begin{align*}
&z=\tau+O(\tau^3)\qquad \text{at}\; z=0, \;\tau=\lambda^{1/2}\\
&\sigma(z)=z+O(z^5)\qquad\text{at}\; z=0\\
\Rightarrow &\log\sigma=\log\tau+ \log(1+O(\tau^2))
\end{align*}
\item $\log\sigma(\omega_3/2)=1/4(\eta_1+\eta_2)(\omega_1+\omega_2)-1/4\log u_1u_2$. This can be obtained from the addition formula
$$
\wp(u)-\wp(v)=-\frac{\sigma(u+v)\sigma(u-v)}{\sigma^2(u)\sigma^2(v)},
$$
evaluating at $(u=\omega_1/2, v=\omega_3/2)$ and at $(u=\omega_2/2, v=\omega_3/2)$, using the periodicity properties of $\sigma$, and finally taking the product and then the logarithm.
\end{itemize}
Substituting in the previous formula, we finally obtain
$$
\text{p.v.}\int_0^{\omega_3}q(z)dz=\frac{1}{4}\left(\log u_1u_2+i\pi+i\pi\frac{\omega_2}{\omega_1}\right)
$$
\\
We conclude that the functions
\begin{equation}
p_1\doteq\frac{1}{2}\log u_1u_2\qquad p_2\doteq\frac{1}{2\pi i}\frac{\omega_2}{\omega_1}
\end{equation}
are flat coordinates of the intersection form \eqref{gsg}.
\begin{Lemma} The formulas
$$
\frac{\partial}{\partial u_1}\omega_1=-\frac{1}{2u_1(u_1-u_2)}(2\eta_1+e_1\omega_1)\qquad\frac{\partial}{\partial u_2}\omega_1=-\frac{1}{2u_2(u_2-u_1)}(2\eta_1+e_2\omega_1)
$$
$$
\frac{\partial}{\partial u_1}\omega_2=-\frac{1}{2u_1(u_1-u_2)}(2\eta_2+e_1\omega_2)\qquad\frac{\partial}{\partial u_2}\omega_2=-\frac{1}{2u_2(u_2-u_1)}(2\eta_2+e_2\omega_2)
$$
hold true.
\end{Lemma}
\begin{proof}
For example the first one; realize the elliptic curve by gluing two copies of the complex plane along cuts drawn from $0$ to $u_1$ and from $u_2$ to $\infty$. Identifying the $a$-cycle with a path in the first sheet encircling the first cut, we have
\begin{align*}
\frac{\partial}{\partial u_1}\omega_1	&=\frac{\partial}{\partial u_1}\oint_a\frac{d\lambda}{2\mu}=\frac{\partial}{\partial u_1}\oint_a\frac{d\lambda}{2\sqrt{\lambda(\lambda-u_1)(\lambda-u_2)}}\\
						&=\frac{1}{2}\oint_a\frac{d\lambda}{2\mu(\lambda-u_1)}=\frac{1}{2}\oint_a\frac{dz}{\wp(z)-e_1}
\end{align*}
The integrand in the last expression is an elliptic function with only a double pole at $\omega_1/2$. Confronting the expansions at $\omega_1/2$, one checks the equality
$$
\frac{1}{\wp(z)-e_1}=\frac{1}{u_1(u_1-u_2)}(\wp(z-\omega_1/2)-e_1)
$$
which yields the result.
\end{proof}
From the lemma we compute
\begin{equation}\label{Jac}
dp_1=\frac{1}{2} \left(\frac{du_1}{u_1}+\frac{du_2}{u_2}\right)\qquad dp_2=\frac{1}{2\omega_1^2(u_1-u_2)}\left(\frac{du_1}{u_1}-\frac{du_2}{u_2}\right)
\end{equation}
We can now easily check the metric takes antidiagonal form in the coordinates $p^a$:
$$
g^{-1}(dp_a, dp_b)=\delta_{a+b, 3}
$$
Inverting \eqref{Jac} we obtain
\begin{equation}\label{invJac}
du_1=u_1dp_1+\omega_1^2u_1(u_1-u_2)dp_2
\end{equation}
$$
du_2=u_2dp_1-\omega_1^2u_2(u_1-u_2)dp_2
$$
The explicit transformation from the flat coordinates of the intersection form to the canonical coordinates can be derived using the identities for Jacobi theta functions\footnote{We adopt the following conventions for the Jacobi theta functions:
\begin{align*}
&\theta_1(z, \tau)=-i\sum_{n=-\infty}^{+\infty}(-1)^nq^{(n+1/2)^2}e^{(2n+1)z}&\qquad&\theta_2(z, \tau)=\sum_{n=-\infty}^{+\infty}q^{(n+1/2)^2}e^{(2n+1)z}\\
&\theta_3(z, \tau)=\sum_{n=-\infty}^{+\infty}q^{n^2}e^{2nz}&\qquad&\theta_4(z, \tau)=\sum_{n=-\infty}^{+\infty}(-1)^nq^{n^2}e^{2nz}
\end{align*}
where $q=e^{i\pi\tau}$ and $\tau = \omega_2/\omega_1.$}
\begin{align*}
	&e_1-e_3=\left(\frac{\pi}{\omega_1}\right)^2\theta_4^4(0, 2\pi ip_2)\\
	&e_2-e_3=-\left(\frac{\pi}{\omega_1}\right)^2\theta_2^4(0, 2\pi i p_2)\\
	&e_1-e_2=\left(\frac{\pi}{\omega_1}\right)^2\theta_3^4(0, 2\pi i p_2)
\end{align*}
which yield
\begin{equation}
u_1=ie^{p_1}\frac{\theta_4^2(0, 2\pi ip_2)}{\theta_2^2(0, 2\pi i p_2)}\qquad u_2=-ie^{p_1}\frac{\theta_2^2(0, 2\pi i p_2)}{\theta_4^2(0, 2\pi i p_2)}
\end{equation}
The structure constants of $\,\star\,$ are computed substituting \eqref{Jac}, \eqref{invJac} in 
$$
\hat{c}_{ab}^c=\sum_{i=1}^2\frac{1}{u_i}\frac{\partial p_c}{\partial u_i}\frac{\partial u_i}{\partial p_a}\frac{\partial u_i}{\partial p_b}
$$
which is just the definition \eqref{2ndssprod}. After lowering one index, we find
\begin{align*}
&\hat{c}_{111}=0&\qquad&\hat{c}_{112}=1\\
&\hat{c}_{122}=0&\qquad&\hat{c}_{222}=\omega_1^4(u_1-u_2)^2=\pi^4\theta_3^8(0, 2\pi ip_2)
\end{align*}
The prepotential reads
\begin{equation}
\hat{F}(p_1, p_2)=\frac{1}{2}p_1^2p_2+f(p_2)
\end{equation}
where $f$ is defined by the equation 
$$
f'''(p)=\pi^4\theta_3^8(0, 2\pi i p)
$$
\section*{Further developments}
An attractive feature of dual-type Frobenius structures on double Hurwitz space lies in their apparent involvement in the equivariant Gromov-Witten theory of local Calabi-Yau threefolds.\\
Brini, Carlet and Rossi showed in \cite{BrCaRo} that the Ablowitz-Ladik hierarchy can be realized as the simplest rational reduction of the 2D-Toda hierarchy. Its dispersionless limit is described by the dual-type structure on the genus zero Hurwitz space with $\mu=\nu=(1, 1)$, and quasi-momentum of the third kind with singularities at a pole and a zero of $\lambda$. It was first pointed out by Brini \cite{Br-Loc} that the corresponding prepotential coincides with the genus zero primary Gromov-Witten potential of the resolved conifold, a fact that led him to formulate a conjectural relation between the all genera Gromov-Witten theory of local $\mathbb{P}^1$ and the full Ablowitz-Ladik hierarchy.\\
The general rational reduction of 2D-Toda naturally defines a bi-graded version of the Ablowitz-Ladik hierarchy, whose dispersionless limit is described by the genus zero double Hurwitz space with $\mu =(M, 1,  \dots, 1), \nu=(N, 1, \dots, 1)$. In this case the prepotential of the dual-type structure is found to reproduce the primary Gromov-Witten invariants of local $\mathbb{P}^1$ with orbifold points of weights $N, M$.\\
This subject is part of a work currently in progress in collaboration with Brini, Carlet and Rossi.

\vspace{45pt}
\center{Stefano Romano}\\
\center{SISSA -- International School for Advanced Studies}\\
\center{Via Bonomea, 265 - 34136 Trieste ITALY}\\
\center{Email: \nolinkurl{sromano@sissa.it}}
\end{document}